\documentclass[11pt]{amsart}
\usepackage[foot]{amsaddr}
\usepackage{setspace}
\usepackage{bbm}
\usepackage{apxproof}
\usepackage{mystyle}
\usepackage{proof}
\usepackage[normalem]{ulem}
\usepackage{amsmath}
\usepackage{amsfonts}
\usepackage{latexsym}
\usepackage[all]{xy}  
\usepackage{tikz}
\usepackage{xcolor}
\usepackage{xfrac}
\usepackage{mathtools}
\usepackage{microtype}
\usepackage{amssymb}
\usepackage{mathrsfs}
\usepackage{stmaryrd}
\usepackage[shortlabels]{enumitem}
\usepackage{booktabs}
\usepackage{float}  
\usepackage{natbib}
\usepackage{hyperref}
\usepackage{scalerel}
\usepackage{orcidlink}
\usepackage[toc,title]{appendix}

\usepackage[nameinlink,capitalise]{cleveref}

\setcitestyle{round}

\definecolor{linkcolor}{RGB}{87,41,40} 

\hypersetup{
    colorlinks=true,
    linkcolor=linkcolor,
    citecolor=linkcolor,   
    urlcolor=linkcolor
}
 
\makeatletter
\newcommand{\nonproof}[1]{\relax}
\newcommand{\vast}{\bBigg@{3}}
\newcommand{\Vast}{\bBigg@{4}}
\newcommand{\VVast}{\bBigg@{5}}

\makeatother
 
\DeclareMathOperator{\unarymost}{\boldsymbol{\mathsf{M}}}

\DeclareMathOperator{\unaryforall}{\boldsymbol{\mathsf{A}}}

\definecolor{defcolor}{RGB}{14,45,97}

\theoremstyle{plain}
\newtheorem{thm}{Theorem}[section]
\newtheorem{lem}[thm]{Lemma}

\newtheorem{cor}[thm]{Corollary}

\usetikzlibrary{calc}

\newtheorem*{prop*}{Proposition}
\newtheorem*{cor*}{Corollary}
\newtheorem*{subclaim*}{Claim}

\theoremstyle{definition}
\newtheorem{defn}[thm]{Definition}

\newtheorem*{rmk*}{Remark}
\newtheorem*{defn*}{Definition}
\newtheorem*{exa*}{Example}
\newtheorem*{fct*}{Fact} 

\theoremstyle{remark}

\DeclareMathSymbol{\Nset}{\mathbin}{AMSb}{"4E}
\DeclareMathSymbol{\Zset}{\mathbin}{AMSb}{"5A}
\DeclareMathSymbol{\Rset}{\mathbin}{AMSb}{"52}
\DeclareMathSymbol{\Qset}{\mathbin}{AMSb}{"51}
\DeclareMathSymbol{\Fset}{\mathbin}{AMSb}{"46}
\DeclareMathSymbol{\Cset}{\mathbin}{AMSb}{"43}
\DeclareMathSymbol{\Kset}{\mathbin}{AMSb}{"4B}

\DeclareMathSymbol{\Sset}{\mathbin}{AMSb}{"53}
\newcommand\ind[1]{\ensuremath{\mathbf{1}_{#1}}}

\DeclareFontFamily{U}{matha}{\hyphenchar\font45}
\DeclareFontShape{U}{matha}{m}{n}{ <-6> matha5 <6-7> matha6 <7-8>
matha7 <8-9> matha8 <9-10> matha9 <10-12> matha10 <12-> matha12 }{}
\DeclareSymbolFont{matha}{U}{matha}{m}{n}

\DeclareFontFamily{U}{mathx}{\hyphenchar\font45}
\DeclareFontShape{U}{mathx}{m}{n}{ <-6> mathx5 <6-7> mathx6 <7-8>
mathx7 <8-9> mathx8 <9-10> mathx9 <10-12> mathx10 <12-> mathx12 }{}
\DeclareSymbolFont{mathx}{U}{mathx}{m}{n}

\DeclareMathDelimiter{\ldbrack} {4}{matha}{"76}{mathx}{"30}
\DeclareMathDelimiter{\rdbrack} {5}{matha}{"77}{mathx}{"38}

\DeclareSymbolFont{CMsymbols}{OMS}{cmsy}{m}{n}
\SetSymbolFont{CMsymbols}{bold}{OMS}{cmsy}{b}{n}
\DeclareMathSymbol{\sim}{\mathrel}{CMsymbols}{"18}

\newlength{\mathfrwidth}
\newsavebox{\mathfrbox}

\usetikzlibrary{arrows,shapes,automata,backgrounds,petri} 
\makeindex

\newcommand{\cal}{\mathcal}

\newcommand{\term}[1]{\textbf{#1}}

\newcommand\xqed[1]{%
  \leavevmode\unskip\penalty9999 \hbox{}\nobreak\hfill
  \quad\hbox{{#1}}}
\newcommand\qeddef{\xqed{${\scriptstyle\spadesuit}$}}
\newcommand\qedthm{\xqed{${\scriptstyle\blacksquare}$}}

\newcommand{\dm}{\ensuremath{}}

\definecolor{dfcond}{RGB}{14,45,97}

\newlength{\depthofsumsign}
\newcommand{\pt}[1]{\textup{({\textcolor[RGB]{14,45,97}{\textsc{$\dm$}#1}})}}

\makeatletter
\renewcommand{\email}[2][]{%
  \ifx\emails\@empty\relax\else{\g@addto@macro\emails{,\space}}\fi%
  \@ifnotempty{#1}{\g@addto@macro\emails{\textrm{(#1)}\space}}%
  \g@addto@macro\emails{#2}%
}
\makeatother

\begin{document}

    \title[Measurable Majorities Are Not Finitely Axiomatizable]{Measurable Majorities Are Not Finitely Axiomatizable}

\author[L.S. Moss]{\orcidlinki{Lawrence S. Moss}{0000-0002-9908-5774}$^\dagger$}
\address[$\dagger$]{Dept.~of Mathematics, Indiana University, Bloomington.}
\email[$\dagger$]{lmoss@iu.edu}

\author[A.P. Pedersen]{\orcidlinki{Arthur Paul Pedersen}{0000-0002-2164-6404}$^{\ddagger\,\ast}$}
\address[$\ddagger$]{Dept.~of Computer Science \& the Intel Investigations Lab, the City College of New York; \break the Graduate Center \& Remote Sensing Earth Systems Institute, the City University of New York.}
\email[$\ddagger$]{appedersen@ccny.cuny.edu}

\maketitle

\begin{abstract}
This theoretical note studies the finite axiomatizability of strict majority reasoning in finite social decision frames. \citet{MossPedersen2026} introduce a coherence criterion that characterizes 
exactly when qualitative majority judgments are representable by a finitely additive measure. The question addressed here is whether that coherence criterion can be replaced, in the finite setting, by any bounded finite fragment. We prove that it cannot. For every $k\ge 1$, we construct a maximal standard frame whose shortest coherence violation has length exactly $2k+2$. Hence there is no uniform finite bound on the incoherence index of social decision frames, resolving Conjecture~5.7 from \citet{MossPedersen2026}. The construction is geometric, in the sense that it proceeds via orthogonality and dimension in rational vector spaces, and self-contained: it isolates a symmetric family of half-sized voting blocs and extends it to a maximal frame in which every shorter balanced obstruction is excluded. Along the explicit infinite sequence of universe sizes obtained in the construction, this also establishes the middle-layer family predicted by Conjecture~B.25 from \citet{MossPedersen2026}. Together with the soundness and completeness theorem for the Moss-Pedersen minimal logic for strict majorities, this establishes that measurable social decision frames are not finitely axiomatizable in that language.
\end{abstract}

\section{Introduction}

In the study of strict majority reasoning within finite electorates, qualitative majority judgments cannot always be represented by a finitely additive probability measure. When such a representation fails to exist, the corresponding social decision frame is  \emph{incoherent}. The minimal complexity of this incoherence --- the length of the shortest sequence of voting blocs required to expose a structural contradiction --- is measured by its \emph{index}.

The study of structural bounds on qualitative probability traces back to the Kraft-Pratt-Seidenberg cancellation conditions \citep{KPS}. Within that framework,  \citet{Fishburn:1996} has investigated the function $f(n)$, which measures the minimal length of cancellation conditions required to guarantee representability for an $n$-element state space, proving that $n-1 \le f(n) \le n+1$ for $n \ge 5$. The incoherence index established in this paper operates as the majoritarian analogue to Fishburn's $f(n)$, extending the analysis of representation bounds from full comparative probability to strict majorities.

This places the problem in the broader tradition of representational
measurement theory. From that standpoint, the central question is not merely
whether qualitative judgments can be assigned numbers, but which structural
conditions make such numerical representation legitimate. Classical
measurement theory studies this question by formulating axioms on qualitative
structures and proving representation theorems that connect those structures to
numerical scales \citep{KrantzLuceSuppesTversky1971,LuceKrantzSuppesTversky1990}.
Scott's linear-inequality approach, the Kraft-Pratt-Seidenberg cancellation
conditions, and later work on qualitative probability all show that
representability can depend on finite configurations whose complexity is not
visible from the surface grammar of the judgments
\citep{Scott1964,KPS,Fishburn:1996,Narens1980}. The present note identifies the corresponding phenomenon for strict majorities:
representability is determined by a coherence scheme whose instances are finite
but whose full force is not finitely exhaustible.

To establish that this index is unbounded (Conjectures 5.7 and B.25 in \citet{MossPedersen2026}), one must construct families of subsets where the shortest logical contradiction requires an arbitrarily large sequence of sets. A combinatorial resolution to this conjecture was recently provided by \citet{Blanco2026}. That proof operates by mapping the winning and losing coalitions of trade-robust simple games into a self-dual selector, leveraging a theorem by \citet{TZ95} on strongly rigid magic squares. Through a padding argument, Blanco establishes the existence of a frame with an index of exactly $2k+2$ for all sufficiently large integers $n \ge k(k+1)$, achieving a highly efficient quadratic bound on the necessary size of the electorate.

Our main contribution in this note is the development of an alternative, purely geometric proof of the unboundedness of the incoherence index. Rather than relying on block designs or external theorems from cooperative game theory, we map the properties of subset selection directly into the geometry of rational vector spaces. This approach allows us to reframe the search for incoherent sequences as an evaluation of linear dependencies within the Boolean hypercube. 

While the combinatorial proof of \citet{Blanco2026} achieves a quadratic scaling of the electorate size and captures every sufficiently large $n$, it relies essentially on Taylor and Zwicker's results on simple games. In contrast, our geometric proof is derived entirely from first principles. We define a highly symmetric base of subsets over a universe sized by central binomial coefficients, explicitly compute its linear span, and use a generic separating hyperplane to construct a maximal frame. Although this geometric construction yields an electorate that scales exponentially with respect to the index, it provides a transparent, self-contained mechanism governing why short balanced sequences are excluded outside a controlled core. In brief,  our results  show that representable qualitative majorities admit no finite structural axiomatization.

We interpret finite axiomatizability in the formal language introduced by
\citet{MossPedersen2026}. In that language, terms are built by Boolean
operations from atomic predicates, while the atomic sentences are of the forms $\unaryforall t$ and $\unarymost t$, expressing respectively that  $t$ is true of the whole universe of discourse and that $t$ is true of a majority. The proof system for that language contains an infinite coherence scheme, indexed by finite sequences of terms.
The construction in the present note shows that the infinitude of this
scheme is inescapable: no finite set of sentences in the Moss-Pedersen language for strict majorities axiomatizes exactly the measurable social decision frames.

\section{Frames, Coherence, and the Incoherence Index}

We study strict majority reasoning using social decision frames. For the remainder of this note, let $k \ge 1$ be a fixed integer. All mathematical constructs are parameterized by $k$ to explicitly track their dependencies.

Let $W_k$ be a finite universe of voters. We assume that the cardinality $|W_k|$ is an even number, and write $|W_k|=2n_k$. A social decision frame is a pair $\mathscr{M}_k = (W_k, \mathcal{M}_k)$, where the designated family $\mathcal{M}_k \subseteq \mathscr{P}(W_k)$ contains distinguished subsets interpreted as voting blocs that form a strict majority. Define:
\[
\mathcal{H}_k
\;\coloneqq\;
\Bigl\{\,A\in \tpow{W_k}: A\notin \mathcal{M}_k
\text{ and } W_k\setminus A\notin \mathcal{M}_k\,\Bigr\}.
\]
The family $\mathcal{H}_k$ corresponds to exact ties.

\begin{defn}[Maximal Standard Frames]\label[defn]{def:maximal_frame}

A frame $\mathscr{M}_k=(W_k,\mathcal M_k)$ is  said to be \term{standard} if
\[
\Bigl\{A\in \tpow{W_k}: |A|>n_k\Bigr\}\subseteq \mathcal M_k
\qquad\mbox{and}\qquad
A\in \mathcal M_k \Longrightarrow |A|\ge n_k .
\]
A standard frame is said to be \term{maximal} if for every $A \in\tpow{W_k}$ with
$|A|=n_k$, exactly one of $A$ and $W_k\setminus A$ belongs to
$\mathcal M_k$.\qeddef
\end{defn}

Thus, by resolving every possible complementary pair, a maximal frame ensures that $\mathcal{H}_k = \emptyset$.

\begin{defn}[Perfectly Balanced Sequences]\label[defn]{def:balanced}
A finite sequence of subsets $A_1, \dots, A_m \subseteq W_k$ is \textbf{perfectly balanced} if it uniformly covers the universe of voters exactly $m/2$ times. Algebraically, this is expressed using the standard binary indicator function $\ind{A}$ as
\[
    \sum_{i=1}^{m} \ind{A_i} \;=\; \frac{m}{2} \, \ind{W_k} \, ,
\]
where $\ind{W_k}$ is the vector of all ones over the universe.\qeddef
\end{defn}

The central property under investigation is \emph{coherence}, a structural condition which ensures that a frame behaves consistently with an underlying finitely additive measure.

\begin{defn}[Coherence and the Incoherence Index]\label[defn]{def:coherence}
A frame $\mathscr{M}_k = (W_k, \mathcal{M}_k)$ is \textbf{coherent} if:
\begin{itemize}[leftmargin=35pt,topsep=15pt,itemsep=10pt,labelsep=10pt]
\item[\pt{c}] For every positive integer $m$ and sequence of sets $A_1,\ldots, A_m \subseteq W_k:$
\item[] If \quad\quad  \pt{c1}\; $A_{i}\in \mathcal{M}_k\cup \mathcal{H}_k$ \quad for each $i=1,\ldots, m$, \qquad and \qquad \pt{c2}\quad  $\displaystyle\frac{m}{2} \ind{W_k} \;\geq \; \displaystyle\sum_{i=1}^{m} \ind{A_{i}}$,
\item[] then\quad \pt{c3}\; $A_i\in \mathcal{H}_k$\qquad\quad\;  for each $i=1,\ldots,m$, \qquad and \qquad \pt{c4}\quad $\displaystyle  \frac{m}{2}\ind{W_k} \;=\; \sum_{i=1}^{m} \ind{A_{i}}$. 
\end{itemize}
The frame $\mathscr{M}_k$ is said to be \textbf{incoherent} if it fails to be coherent.\qeddef
\end{defn}
Here and below, inequalities between vectors in $\mathbb{Q}^{W_k}$ are understood pointwise. Thus, condition \pt{c2} dictates that no voter is covered by more than half of the sets in the sequence. If this condition holds, coherence requires that the sequence must be perfectly balanced \pt{c4} and consist entirely of exact ties \pt{c3}.

In a maximal frame, there are no exact ties since $\mathcal{H}_k=\emptyset$. Therefore, for any non-empty sequence satisfying \pt{c1}, condition \pt{c3} cannot hold. Consequently, a maximal frame is coherent if and only if no non-empty sequence of sets drawn from $\mathcal{M}_k$ satisfies the coverage bound \pt{c2}. If such a sequence exists, the frame is incoherent. The minimal length $m$ of such a sequence violating coherence is defined to be \textbf{incoherence index} of the frame.

\section{Bipolar Indicator Vectors and Zero-Sum Conditions}

To analyze sequences of subsets algebraically, we map each subset $A \subseteq W_k$ to a bipolar indicator vector $\mathbf{x}_A \in \{-1,1\}^{W_k}$. We define $\mathbf{x}_A(v) = 1$ if $v \in A$, and $\mathbf{x}_A(v) = -1$ if $v \notin A$. This translates the standard binary indicator function $\ind{A}$ via the affine transformation
\begin{equation} \label{eq:affine}
    \mathbf{x}_A \;=\; 2 \ind{A} - \ind{W_k} \, .
\end{equation}

Two vectors $\mathbf{x}, \mathbf{y} \in \mathbb{Q}^{W_k}$ form an \textbf{antipodal pair} if $\mathbf{y} = -\mathbf{x}$. In this geometry, the affine transformation \eqref{eq:affine} ensures that complementary subsets $A$ and $W_k \setminus A$ correspond precisely to an antipodal pair of vectors $\mathbf{x}_A$ and $-\mathbf{x}_A$. 

A sequence of sets $A_1,\dots,A_m\subseteq W_k$ is perfectly balanced if and
only if the corresponding bipolar indicator vectors sum to $\mathbf{0}$ in
$\mathbb{Q}^{W_k}$. Summing the affine transformation \eqref{eq:affine} over the sequence yields
\[
    \sum_{i=1}^{m} \mathbf{x}_{A_i} 
    \;=\; \sum_{i=1}^{m} \bigl( 2 \ind{A_i} - \ind{W_k} \bigr) 
    \;=\; 2 \Biggl( \, \sum_{i=1}^{m} \ind{A_i} \Biggr) - m \ind{W_k} \, .
\]
Setting the right-hand side to $\mathbf{0}$ is algebraically equivalent to the condition $\sum_{i=1}^m \ind{A_i} = \frac{m}{2} \ind{W_k}$. We express this zero-sum condition as
\[
    \sum_{i=1}^{m} \mathbf{x}_{A_i} \;=\; \mathbf{0} \, .
\]
This zero-sum condition is the point of contact with cooperative game theory. In that literature, a sequence satisfying this exact balance condition for a length of $2j$ is formally known as a $j$-trade, or a $j$-balanced sequence over a self-dual selector \citep{Blanco2026}. Where combinatorial approaches leverage magic-square games to construct these trades, our framework evaluates their linear dependencies by mapping them directly into zero-sum bipolar vectors.

Let $X_k \subset \{-1, 1\}^{W_k}$ be the set of all bipolar indicator vectors whose components sum to zero:
\[
    X_k \;=\; \Biggl\{ \, \mathbf{x} \in \{-1,1\}^{W_k} \;:\; \sum_{v \in W_k} \mathbf{x}(v) = 0 \, \Biggr\} \, .
\]
Thus $X_k$ is the bipolar encoding of the middle layer
\[
\Bigl\{\,A\in \tpow{W_k}: |A|=n_k\,\Bigr\}.
\]
Equivalently, vectors in $X_k$ correspond precisely to subsets of $W_k$ of size exactly $n_k=|W_k|/2$.

\section{The Core Construction}

We explicitly construct the components of a frame whose incoherence index requires a sequence of length $2k+2$. Let the base set of elements be $V_k = \{1, \dots, 2k+2\}$. Define the universe of voters $W_k$ to be the set of all $(k+1)$-element subsets of $V_k$:
\[
    W_k \;=\; \Bigl\{ \, v \subseteq V_k : |v| = k+1 \, \Bigr\} \, .
\]
The total number of voters is $|W_k|=\binom{2k+2}{k+1}$. This central binomial coefficient is even for $k\ge 1$: the complement map
\[
v\longmapsto V_k\setminus v
\]
is a fixed-point-free involution on the set of $(k+1)$-element subsets of $V_k$. Hence $|W_k|$ is even. We denote
\[
n_k=\frac{1}{2}\binom{2k+2}{k+1}.
\]

\begin{defn}[Dictator Blocs]
For each base element $j \in V_k$, we define a \textbf{dictator bloc} 
$M_{k,j} \subseteq W_k$
consisting of all voters who possess the element $j$:
\[
    M_{k,j} \;=\; \Bigl\{ \, v \in W_k : j \in v \, \Bigr\} \, .
\]
Indeed, fixing $j\in V_k$, a voter $v\in W_k$ belongs to $M_{k,j}$ exactly when $v=\{j\}\cup s$, where $s\subseteq V_k\setminus\{j\}$ and $|s|=k$. Hence:
\[
|M_{k,j}|
=
\binom{2k+1}{k}
=
\frac{1}{2}\binom{2k+2}{k+1}
=
n_k.
\]
Thus, each $M_{k,j}$ contains exactly half of the voters.\qeddef
\end{defn}

Every voter $v \in W_k$ is a subset of $V_k$ containing exactly $k+1$ elements. Therefore, every voter belongs to exactly $k+1$ of the sets in the sequence $M_{k,1}, \dots, M_{k,2k+2}$.
The sequence has length $m = 2k+2$, 
meaning the coverage for every voter is exactly $m/2$. Consequently, this specific sequence of subsets is perfectly balanced.

\begin{defn}[The Core Set]
Let $\mathbf{y}_{k,j} = \mathbf{x}_{M_{k,j}} \in X_k$ be the bipolar indicator vector corresponding to $M_{k,j}$. Because the sequence of blocs is perfectly balanced, the corresponding vectors satisfy the zero-sum condition
\[
    \sum_{j=1}^{2k+2} \mathbf{y}_{k,j} \;=\; \mathbf{0} \, .
\] 
We define the \textbf{core set} $\mathcal{C}_k$ as the set of these $2k+2$ balanced vectors:
\[
    \mathcal{C}_k \;=\; \Bigl\{ \mathbf{y}_{k,1}, \dots, \mathbf{y}_{k,2k+2} \Bigr\} \, .
\]
\qeddef
\end{defn}

\section{Algebraic Properties of the Core Sequence}

\begin{lem}[Minimality]\label[lem]{lem:minimality}
Any non-empty perfectly balanced sequence formed by drawing sets exclusively from $M_{k,1}, \dots, M_{k,2k+2}$ must have a length that is a multiple of $2k+2$.\qedthm
\end{lem}
\begin{proof}
Let $c_j \ge 0$ denote the integer number of times $M_{k,j}$ appears in the sequence. By hypothesis, the sequence is perfectly balanced, meaning the corresponding bipolar vectors sum to zero in $\mathbb{Q}^{W_k}$:
\[
    \sum_{j=1}^{2k+2} c_j \mathbf{y}_{k,j} \;=\; \mathbf{0} \, .
\]
Let $m = \sum_{j=1}^{2k+2} c_j$ be the total length of the sequence. Evaluating the sum at an arbitrary voter $v \in W_k$ yields
\[
    \sum_{j=1}^{2k+2} c_j \mathbf{y}_{k,j}(v) \;=\; 0 \, .
\]
Substituting $\mathbf{y}_{k,j}(v) = 2 \ind{M_{k,j}}(v) - 1$ using \eqref{eq:affine}, we obtain
\[
    \sum_{j=1}^{2k+2} c_j \bigl( 2 \ind{M_{k,j}}(v) - 1 \bigr) \;=\; 0 \, .
\]
Distributing the summation yields
\[
    2 \Biggl( \, \sum_{j=1}^{2k+2} c_j \ind{M_{k,j}}(v) \Biggr) - \sum_{j=1}^{2k+2} c_j \;=\; 0 \, .
\]
Since a voter $v$ belongs to $M_{k,j}$ if and only if $j \in v$, the indicator evaluates to $\ind{M_{k,j}}(v) = 1$ strictly when $j \in v$. Substituting the sequence length $\sum_{j=1}^{2k+2} c_j = m$, we find
\[
    2 \Biggl( \, \sum_{j \in v} c_j \Biggr) - m \;=\; 0 \qquad \implies \qquad \sum_{j \in v} c_j \;=\; \frac{m}{2} \, .
\]
This demonstrates that the sum of the coefficients corresponding to the elements inside any voter $v$ must equal the constant $m/2$.

Let $a$ and $b$ be two distinct elements in $V_k$. We construct two adjacent voters $v_a$ and $v_b$ in $W_k$. Choose a subset $S_0 \subset V_k \setminus \{a,b\}$ of size exactly $k$. This is possible because $|V_k \setminus \{a,b\}| = 2k \ge k$. Define $v_a = S_0 \cup \{a\}$ and $v_b = S_0 \cup \{b\}$. Both $v_a$ and $v_b$ have cardinality $k+1$ and are therefore valid voters in $W_k$.

The balanced sum condition requires
\[
    \sum_{j \in v_a} c_j \;=\; \frac{m}{2} \qquad \text{and} \qquad \sum_{j \in v_b} c_j \;=\; \frac{m}{2} \, .
\]
Subtracting the two equations eliminates the shared elements in $S_0$:
\[
    \Biggl( c_a + \sum_{j \in S_0} c_j \Biggr) - \Biggl( c_b + \sum_{j \in S_0} c_j \Biggr) \;=\; 0 \qquad \implies \qquad c_a \;=\; c_b \, .
\]
Since the indices $a$ and $b$ are arbitrary, all coefficients $c_j$ must equal a single uniform constant $c \ge 0$. The total length of the sequence is therefore given by
\[
    m \;=\; \sum_{j=1}^{2k+2} c_j \;=\; \sum_{j=1}^{2k+2} c \;=\; c(2k+2) \, .
\]
Since the sequence is non-empty, we must have $c \ge 1$. Thus every non-empty perfectly balanced sequence drawn from the core consists of exactly $c$ copies of each dictator bloc $M_{k,j}$. In particular, its length is a multiple of $2k+2$, with the minimal non-zero length being exactly $2k+2$.
\end{proof}

Alternatively, \cref{lem:minimality} may be reformulated as follows.

\begin{lem}  
Let $S_1, \ldots, S_r$ be a non-empty sequence of sets, and suppose that for each $i$ there is a number $n(i)$ such that $S_i = M_{n(i)}$.
Assume that $S_1, \ldots, S_k$ is perfectly balanced.
Then $r \geq m$.
\end{lem}

\begin{proof} First, fix a voter $v\in W$.
 We begin by determining the size of the set $Z$, where 
 \[ Z_v = \set{(i,j) \colon 1\leq i \leq r,  j\in v, \mbox{ and } n(i) = j}\] in two ways.
On the one hand, 
\[|Z_v| = 
\sum_{j\in v} |\set{i \leq r : n(i) = j}|.
\]
On the other, $Z_v = \set{i : n(i) \in v}  = \set{i : v \in S_{n(i)}}$.
As a result, we see that 
\begin{equation}\label{as-a-result-one}
|Z| = 
\sum_{j\in v} |\set{i \leq r : n(i) = j}| = \frac{m}{2}.
\end{equation}
This holds for all $v\in W$.  
As a result, $|Z_v|$ is independent of $v$.  For each $j \leq m$, let $c_j =  |\set{i \leq r : n(i) = j}|$.
Then from our work above, we see that for all $v$, $\sum_{j\in v} c_j = \frac{m}{2}$.

Now let $a, b\in [m]$; we show that $c_a = c_b$.
Let $v_0 \subseteq [m]$ be any set such that $v_0 \cup\set{a}$ and $v_0 \cup\set{b}$ belong to $W$.
(Such a set $v_0$ exists since $m \geq 4$.)
Then 
\[
\frac{m}{2} = \sum_{j\in v_0 \cup \set{a}} c_j = c_a + \sum_{j\in v_0} c(j)
\]
Similarly, $\frac{m}{2} = c_b + \sum_{j\in v_0} c_j$.  It follows that $c(a) = c(b)$, as claimed.

So the function $a\mapsto c_a$ is constant on $[m]$.  The value of this function cannot be $0$, since 
$\sum_{a\in [m]} c_a = r > 0$.
So each number $c_a$ is at least $1$.  Hence $r = \sum_{a\in [m]} c_a > |[m]| = m$.
\end{proof}

\begin{lem}[Linear Intersection Lemma]\label[lem]{lem:span}
Let $L_k = \mathrm{span}_{\mathbb{Q}}(\mathcal{C}_k)$ be the linear span of $\mathcal{C}_k$ over the rational numbers. The only bipolar indicator vectors representing sets of size $n_k$ that lie in $L_k$ are precisely the core vectors and their antipodes. We write this as:
\[
    L_k \cap X_k \;=\; \Bigl\{\, \pm \mathbf{y}_{k,1}, \dots, \pm \mathbf{y}_{k,2k+2}\, \Bigr\} \, .
\]\qedthm
\end{lem}
\begin{proof}
Let $\mathbf{x} = \sum_{j=1}^{2k+2} \alpha_j \mathbf{y}_{k,j} \in L_k \cap X_k$. For any voter $v \in W_k$, we apply the affine transformation to express the component $\mathbf{x}(v)$:
\[
    \mathbf{x}(v) 
    \;=\; \sum_{j=1}^{2k+2} \alpha_j \mathbf{y}_{k,j}(v) 
    \;=\; \sum_{j=1}^{2k+2} \alpha_j \bigl( 2\ind{M_{k,j}}(v) - 1 \bigr) 
    \;=\; 2 \Biggl( \, \sum_{j \in v} \alpha_j \Biggr) - \sum_{j=1}^{2k+2} \alpha_j \, .
\]
We define $\beta_j = 2\alpha_j$ and let the total sum be $S = \sum_{j=1}^{2k+2} \alpha_j$. This yields
\[
    \mathbf{x}(v) \;=\; \sum_{j \in v} \beta_j - S \, .
\]
Because $\mathbf{x} \in X_k$, its components must satisfy $\mathbf{x}(v) \in \{-1,1\}$. We require
\[
    \sum_{j \in v} \beta_j - S \;\in\; \bigl\{\,-1, \, 1\,\bigr\} \qquad \implies \qquad \sum_{j \in v} \beta_j \;\in\; \bigl\{\, S-1, \, S+1 \,\bigr\} \quad \text{for all } v \in W_k \, .
\]
We denote this sum by $S(v) = \sum_{j \in v} \beta_j$.

We utilize the adjacent voters $v_a = S_0 \cup \{a\}$ and $v_b = S_0 \cup \{b\}$ constructed in Lemma \ref{lem:minimality}. The difference $S(v_a) - S(v_b)$ evaluates to
\[
    S(v_a) - S(v_b) \;=\; \Biggl( \beta_a + \sum_{j \in S_0} \beta_j \Biggr) - \Biggl( \beta_b + \sum_{j \in S_0} \beta_j \Biggr) \;=\; \beta_a - \beta_b \, .
\]
Since both $S(v_a)$ and $S(v_b)$ belong to the set $\{S-1, S+1\}$, their absolute difference is at most $2$. Thus, we must have
\[
    \beta_a - \beta_b \;\in\; \{-2, \, 0, \, 2\} \, .
\]
It follows that the values $\beta_j$ can take at most two distinct numerical values over the index set $V_k$:
Indeed, if three distinct values occurred, then the largest and smallest would differ by at least $4$, contradicting the fact that every pairwise difference belongs to $\{-2,0,2\}$.

If all $\beta_j$ were identically equal, then $S(v)$ would be constant for all $v \in W_k$, forcing $\mathbf{x}$ to be a constant vector. However, since $\mathbf{x} \in X_k$, its components sum to zero. As $W_k$ is non-empty ($\binom{2k+2}{k+1} \ge 6$), a constant vector summing to zero must be the zero vector $\mathbf{0}$. This contradicts $\mathbf{x}(v) \in \{-1, 1\}$. Therefore, the elements $\beta_j$ must take exactly two distinct numerical values, and these values must differ by exactly $2$.

Let these two values be $\lambda$ and $\lambda - 2$. Let $p$ be the integer number of components equal to $\lambda$. The remaining $2k+2-p$ components equal $\lambda - 2$. We sort the sequence of coefficients in descending order, writing $\beta_{(1)} \ge \dots \ge \beta_{(2k+2)}$. The first $p$ elements of this sorted sequence equal $\lambda$, and the remaining elements equal $\lambda - 2$.

Because $W_k$ consists of all subsets of $V_k$ of size $k+1$, there exists a voter $v_{\max} \in W_k$ containing the elements corresponding to the $k+1$ largest values of $\beta_j$, and another voter $v_{\min} \in W_k$ containing the $k+1$ smallest values. Thus, the maximum possible value of $S(v)$ over $W_k$ is the sum of the $k+1$ largest values, and the minimum is the sum of the $k+1$ smallest values. Since $S(v)$ must fall into $\{S-1, S+1\}$ and is not constant, the maximum sum must attain $S+1$ and the minimum sum must attain $S-1$. We strictly require the difference between the maximum sum and the minimum sum to be exactly $2$:
\[
    S(v_{\max}) - S(v_{\min}) \;=\; \sum_{i=1}^{k+1} \beta_{(i)} - \sum_{i=k+2}^{2k+2} \beta_{(i)} \;=\; \sum_{i=1}^{k+1} \bigl( \beta_{(i)} - \beta_{(i+k+1)} \bigr) \;=\; 2 \, .
\]

The difference term $\beta_{(i)} - \beta_{(i+k+1)}$ must evaluate to either $0$ or $2$. Exactly one term in the summation equals $2$, while all other terms equal $0$. This requires exactly one index $i^*$ where $\beta_{(i^*)} = \lambda$ and $\beta_{(i^*+k+1)} = \lambda - 2$. 

Due to the sorted descending order, $\beta_{(i^*)} = \lambda$ implies $i^* \le p$. Simultaneously, $\beta_{(i^*+k+1)} = \lambda - 2$ implies $i^*+k+1 > p$. We combine this condition to
\[
    p-k \;\le\; i^* \;\le\; p \, .
\]

The index $i^*$ is constrained by the bounds of the summation, meaning $1 \le i^* \le k+1$. Therefore, the index $i^*$ must satisfy
\[
    \max(1, \, p-k) \;\le\; i^* \;\le\; \min(p, \, k+1) \, .
\]
For $i^*$ to be uniquely determined (as required by the exact sum of $2$), the number of valid integers in this interval must exactly equal $1$. We analyze the length of this discrete interval:
\[
    \min(p, \, k+1) - \max(1, \, p-k) + 1 \;=\; 1 \, .
\]
We evaluate this equation over the full domain $1 \le p \le 2k+1$:
\begin{itemize}
    \item[] \textbf{Case 1 ($1 \le p \le k$):} The length simplifies to $p - 1 + 1 = p$. Setting this to $1$ yields $p=1$.
    \item[] \textbf{Case 2 ($p = k+1$):} The length simplifies to $(k+1) - 1 + 1 = k+1$. Since $k \ge 1$, we have $k+1 \ge 2 \neq 1$. No solution exists in this range.
    \item[] \textbf{Case 3 ($k+2 \le p \le 2k+1$):} The length simplifies to $(k+1) - (p-k) + 1 = 2k+2-p$. Setting this to $1$ yields $p=2k+1$.
\end{itemize}
Thus, there are exactly two integer solutions: $p=1$ and $p=2k+1$.

Suppose first that $p=1$, and let $j_0$ be the unique index with
$\beta_{j_0}=\lambda$. Then
\[
S
=
\sum_{j=1}^{2k+2}\alpha_j
=
\frac{1}{2}\sum_{j=1}^{2k+2}\beta_j
=
\frac{1}{2}\Bigl(\lambda+(2k+1)(\lambda-2)\Bigr)
=
(k+1)\lambda-(2k+1).
\]
If $j_0\in v$, then
\[
\sum_{j\in v}\beta_j
=
\lambda+k(\lambda-2)
=
(k+1)\lambda-2k,
\]
and hence $\mathbf{x}(v)=1$. If $j_0\notin v$, then
\[
\sum_{j\in v}\beta_j
=
(k+1)(\lambda-2)
=
(k+1)\lambda-2k-2,
\]
and hence $\mathbf{x}(v)=-1$. Therefore
\[
\mathbf{x}(v)
=
2\ind{M_{k,j_0}}(v)-1
=
\mathbf{y}_{k,j_0}(v)
\]
for every $v\in W_k$, so $\mathbf{x}=\mathbf{y}_{k,j_0}$.

Suppose next that $p=2k+1$, and let $j_0$ be the unique index with
$\beta_{j_0}=\lambda-2$. Then
\[
S
=
\sum_{j=1}^{2k+2}\alpha_j
=
\frac{1}{2}\sum_{j=1}^{2k+2}\beta_j
=
\frac{1}{2}\Bigl((2k+1)\lambda+(\lambda-2)\Bigr)
=
(k+1)\lambda-1.
\]
If $j_0\in v$, then
\[
\sum_{j\in v}\beta_j
=
(\lambda-2)+k\lambda
=
(k+1)\lambda-2,
\]
and hence $\mathbf{x}(v)=-1$. If $j_0\notin v$, then
\[
\sum_{j\in v}\beta_j
=
(k+1)\lambda,
\]
and hence $\mathbf{x}(v)=1$. Therefore
\[
\mathbf{x}(v)
=
1-2\ind{M_{k,j_0}}(v)
=
-\mathbf{y}_{k,j_0}(v)
\]
for every $v\in W_k$, so $\mathbf{x}=-\mathbf{y}_{k,j_0}$.

The intersection $L_k\cap X_k$ contains no other vectors.
\end{proof}

\section{A Generic Vector Avoiding Finitely Many Hyperplanes}

We recall a classic result.

\begin{lem}[Finite-Union Lemma, \cite{BB59}]\label[lem]{lem:finite_union}
Let $K$ be an infinite field, and let $V$ be a finite-dimensional vector space
over $K$. If $U_1,\dots,U_r$ are proper linear subspaces of $V$, then
\[
V\neq \bigcup_{\ell=1}^r U_\ell .
\]\qedthm
\end{lem}

\begin{proof}
We argue by induction on $r$. The case $r=1$ is immediate. Suppose
\[
V=\bigcup_{\ell=1}^r U_\ell
\]
with each $U_\ell$ proper, and assume $r$ is minimal. Then none of the
$U_\ell$ is contained in the union of the others. Choose
\[
\mathbf{a}\in U_1\setminus \bigcup_{\ell=2}^r U_\ell
\qquad\text{and}\qquad
\mathbf{b}\in V\setminus U_1 .
\]
For each $t\in K$, set $\mathbf{v}_t=\mathbf{a}+t\mathbf{b}$. Since
$\mathbf{b}\notin U_1$, the affine line $\{\mathbf{v}_t:t\in K\}$ meets
$U_1$ only at $t=0$. For each $\ell\ge 2$, the set of $t\in K$ such that
$\mathbf{v}_t\in U_\ell$ has at most one element: if
$\mathbf{a}+t\mathbf{b}$ and $\mathbf{a}+s\mathbf{b}$ both belong to
$U_\ell$ with $t\neq s$, then $\mathbf{b}\in U_\ell$, and hence
$\mathbf{a}\in U_\ell$, contradicting the choice of $\mathbf{a}$. Thus the
line contains infinitely many points but meets the finite union
$\bigcup_{\ell=1}^r U_\ell$ in only finitely many points, a contradiction.
\end{proof}

We use \label[lem]{lem:finite_union} to establish the following lemma. 

\begin{lem}[Generic Vector Avoiding Finitely Many Hyperplanes]\label[lem]{lem:hyperplane}
There exists a rational vector $\mathbf{u}^*_k \in \mathbb{Q}^{W_k}$ possessing the following three properties:
\begin{enumerate}
    \item[\textup{(1)}] $\displaystyle \sum_{v \in W_k} \mathbf{u}^*_k(v) = 0$.
    \item[\textup{(2)}] $\mathbf{u}^*_k \cdot \mathbf{y}_{k,j} = 0$ \quad for all $j \in \{1, \dots, 2k+2\}$.
    \item[\textup{(3)}] $\mathbf{u}^*_k \cdot \mathbf{x} \neq 0$ \quad for all $\mathbf{x} \in X_k \setminus L_k$.
\end{enumerate}\qedthm
\end{lem}
\begin{proof}
Let $Z_k$ be the subspace of vectors in $\mathbb{Q}^{W_k}$ whose components sum to zero:
\[
    Z_k \;=\; \Biggl\{ \, \mathbf{u} \in \mathbb{Q}^{W_k} : \sum_{v \in W_k} \mathbf{u}(v) = 0 \, \Biggr\} \, .
\]
Since $W_k$ has size $2n_k=\binom{2k+2}{k+1}$, the dimension of $Z_k$ is
\[
\dim(Z_k)=\binom{2k+2}{k+1}-1.
\]
Each vector in $\mathcal{C}_k$ belongs to $X_k$, and hence to $Z_k$. Moreover, the core relation
\[
\sum_{j=1}^{2k+2}\mathbf{y}_{k,j}=\mathbf{0}
\]
gives one non-trivial linear dependence among the $2k+2$ generators. Therefore
\[
\dim(L_k)\le (2k+2)-1=2k+1.
\]
Let
\[
U_k=Z_k\cap L_k^\perp .
\]
Since $L_k\subseteq Z_k$ and the standard dot product restricts non-degenerately to $Z_k$, this is the orthogonal complement of $L_k$ inside $Z_k$, and
\[
\dim(U_k)=\dim(Z_k)-\dim(L_k).
\]
Consequently,
\[
\dim(U_k)
\ge
\binom{2k+2}{k+1}-1-(2k+1)
=
\binom{2k+2}{k+1}-2k-2.
\]
For $k=1$, this lower bound is $6-4=2$. For $k\ge 2$, it is at least $20-6=14$, and hence is also at least $2$. Thus $\dim(U_k)\ge 2$ for all $k\ge 1$.

The set $X_k \setminus L_k$ is finite. Each vector $\mathbf{x}$ in this set defines an orthogonal constraint $\mathbf{u} \cdot \mathbf{x} = 0$. For each $\mathbf{x}\in X_k\setminus L_k$, the functional
\[
\mathbf{u}\mapsto \mathbf{u}\cdot \mathbf{x}
\]
is not identically zero on $U_k$. Indeed, if $\mathbf{u}\cdot \mathbf{x}=0$ for every
$\mathbf{u}\in U_k$, then $\mathbf{x}\in U_k^\perp\cap Z_k=L_k$, since
$Z_k=L_k\oplus U_k$ orthogonally inside $Z_k$, a contradiction.
Therefore
\[
    U_{\mathbf{x}} \;=\; \Bigl\{ \, \mathbf{u} \in U_k : \mathbf{u} \cdot \mathbf{x} = 0 \, \Bigr\}
\]
is a proper hyperplane of $U_k$.

By \cref{lem:finite_union}, the rational vector space $U_k$ is not the
union of the finitely many proper hyperplanes
\[
\bigl\{\,U_{\mathbf{x}}:\mathbf{x}\in X_k\setminus L_k\,\bigr\}.
\]
Choose
\[
\mathbf{u}^*_k\in U_k\setminus \bigcup_{\mathbf{x}\in X_k\setminus L_k}U_{\mathbf{x}}.
\]

Since $\mathbf{u}^*_k \in U_k \subset Z_k$, its components sum to zero (Property 1). Since $\mathbf{u}^*_k \in L_k^\perp$, it is orthogonal to all core vectors $\mathbf{y}_{k,j}$ (Property 2). Since $\mathbf{u}^*_k$ avoids the hyperplanes defined by $\mathbf{x} \in X_k \setminus L_k$, the dot product $\mathbf{u}^*_k \cdot \mathbf{x}$ is non-zero for all such vectors (Property 3).
\end{proof}

\section{Resolution of the Conjecture}

\begin{thm}\label[thm]{thm:main}
For any integer $k \ge 1$, there exists a maximal standard frame $\mathscr{M}_k = (W_k, \mathcal{M}_k)$ whose incoherence index is exactly $2k+2$. Moreover, the length-$2k+2$ witness is the core sequence
\[
M_{k,1},\dots,M_{k,2k+2}.
\]\qedthm
\end{thm}
\begin{proof}
We use the hyperplane-avoiding vector $\mathbf{u}^*_k$ from \cref{lem:hyperplane} to define a tie-breaking family of bipolar indicator vectors $\mathcal{F}_k$: 
\[
    \mathcal{F}_k \;=\; \Bigl\{ \, \mathbf{x} \in X_k \setminus L_k : \mathbf{u}^*_k \cdot \mathbf{x} > 0 \, \Bigr\} \;\cup\; \mathcal{C}_k \, .
\]
Here $\mathcal{F}_k$ is a family of bipolar vectors. The corresponding family
of middle-layer subsets is
\[
\mathcal{F}^{\mathrm{set}}_k
=
\Bigl\{\,A\in \tpow{W_k}:|A|=n_k\text{ and }\mathbf{x}_A\in\mathcal{F}_k\,\Bigr\}.
\]

For any $\mathbf{x} \in X_k \setminus L_k$, exactly one of the vectors $\pm \mathbf{x}$ yields a strictly positive dot product with $\mathbf{u}^*_k$. Exactly one is included in $\mathcal{F}_k$. For vectors inside $L_k \cap X_k$, \cref{lem:span} proved these are exactly the antipodal pairs $\pm \mathbf{y}_{k,j}$. If $a\neq b$, choose $v\in W_k$ with $a,b\in v$, which is possible since
$k+1\ge 2$. Then
\[
\mathbf{y}_{k,a}(v)=1
\qquad\text{and}\qquad
\mathbf{y}_{k,b}(v)=1,
\]
so $\mathbf{y}_{k,a}(v)\neq -\mathbf{y}_{k,b}(v)$. Hence
$\mathbf{y}_{k,a}\neq -\mathbf{y}_{k,b}$ whenever $a\neq b$. We manually include $\mathbf{y}_{k,j} \in \mathcal{C}_k$ and exclude $-\mathbf{y}_{k,j}$. Thus, the set $\mathcal{F}_k$ precisely contains exactly one vector from every antipodal pair in $X_k$.

We define the frame $\mathscr{M}_k = (W_k, \mathcal{M}_k)$ by setting:
\begin{align*}
\mathcal M_k
\quad&\coloneqq\quad
\Bigl\{\,A\in \tpow{W_k}: |A|>n_k\,\Bigr\}
\;\cup\;
\Bigl\{\,A\in \tpow{W_k}: |A|=n_k \text{ and } \mathbf{x}_A\in \mathcal{F}_k\,\Bigr\}.\\
&=\quad\Bigl\{\,A\in \tpow{W_k}: |A|>n_k\,\Bigr\}
\;\cup\;
\mathcal{F}^{\mathrm{set}}_k .
\end{align*}
Thus every member of $\mathcal{M}_k$ has size at least $n_k$, and every subset of $W_k$ of size strictly greater than $n_k$ belongs to $\mathcal{M}_k$. On the middle layer, membership is determined by $\mathcal{F}_k$. Because $\mathcal{F}_k$ contains exactly one vector from every antipodal pair in $X_k$, exactly one of $A$ and $W_k\setminus A$ belongs to $\mathcal{M}_k$ for every $A\subseteq W_k$ with $|A|=n_k$. Hence $\mathscr{M}_k$ is a maximal standard frame, and $\mathcal{H}_k=\emptyset$.

Since $\mathcal{H}_k=\emptyset$, any non-empty sequence $A_1,\dots,A_m$ with $A_i\in\mathcal{M}_k$ for all $i$ violates coherence exactly when it satisfies the coverage bound \pt{c2}. Indeed, \pt{c1} then holds automatically, while \pt{c3} is impossible. We write the coverage bound as
\[
    \sum_{i=1}^{m} \ind{A_i}(v) \;\le\; \frac{m}{2} \qquad \text{for all } v \in W_k \, .
\]
We sum the sizes of these sets, evaluating over all voters: 
\[
    \sum_{v \in W_k} \sum_{i=1}^{m} \ind{A_i}(v) \;\le\; \sum_{v \in W_k} \frac{m}{2} \;=\; |W_k| \, \frac{m}{2} \;=\; 2n_k \Bigl(\frac{m}{2}\Bigr) \;=\; m n_k \, .
\]
By reversing the order of summation, we double-count the total sizes of the sets: 
\[
    \sum_{i=1}^{m} \sum_{v \in W_k} \ind{A_i}(v) \;=\; \sum_{i=1}^{m} |A_i| \, .
\]
This yields $\sum_{i=1}^m |A_i|\le mn_k$. Since every member of $\mathcal{M}_k$ has size at least $n_k$, we also have $\sum_{i=1}^m |A_i|\ge mn_k$. Hence equality holds in both estimates. It follows that $|A_i|=n_k$ for every $i$, and that the pointwise inequalities in \pt{c2} are all equalities:
\[
    \sum_{i=1}^{m} \ind{A_i}(v) \;=\; \frac{m}{2} \qquad \text{for all } v \in W_k \, .
\]
Any sequence violating coherence must therefore be a perfectly balanced sequence of size-$n_k$ sets.

We translate this perfectly balanced sequence to bipolar indicators using equation \eqref{eq:affine}, $\mathbf{x}_{A_i} = 2\ind{A_i} - \ind{W_k}$:
\begin{align*}
    \sum_{i=1}^{m} \mathbf{x}_{A_i} 
    &\;=\; \sum_{i=1}^{m} \bigl( 2\ind{A_i} - \ind{W_k} \bigr) \\[1ex]
    &\;=\; 2 \Biggl( \, \sum_{i=1}^{m} \ind{A_i} \Biggr) - m \ind{W_k} \\[1ex]
    &\;=\; 2 \Biggl( \, \frac{m}{2} \ind{W_k} \Biggr) - m \ind{W_k} \\[1ex]
    &\;=\; \mathbf{0} \, .
\end{align*}

Each $\mathbf{x}_{A_i}$ corresponds to a set of size $n_k$ in $\mathcal{M}_k$, so $\mathbf{x}_{A_i} \in \mathcal{F}_k$. We take the dot product of this zero-sum condition with $\mathbf{u}^*_k$:
\[
    \sum_{i=1}^{m} \bigl( \mathbf{u}^*_k \cdot \mathbf{x}_{A_i} \bigr) \;=\; 0 \, .
\]
By construction of $\mathcal{F}_k$, we have $\mathbf{u}^*_k\cdot\mathbf{x}>0$ for every $\mathbf{x}\in\mathcal{F}_k\setminus\mathcal{C}_k$, while $\mathbf{u}^*_k\cdot\mathbf{x}=0$ for every $\mathbf{x}\in\mathcal{C}_k$. Thus $\mathbf{u}^*_k\cdot\mathbf{x}\ge 0$ for all $\mathbf{x}\in\mathcal{F}_k$. Since a sum of non-negative rational numbers is zero only when every term is zero, we must have
\[
\mathbf{u}^*_k\cdot\mathbf{x}_{A_i}=0
\qquad\text{for every }i.
\]
By \cref{lem:hyperplane}, this implies $\mathbf{x}_{A_i}\in L_k$ for every $i$.
\footnote{I feel that there's a missing line here, when we go back from the $\mathbf{x}$s to the $\mathbf{y}$s.
That is, this ending is too fast, especially considering the meticulous detail in the rest of the proof.}
By \cref{lem:span}, the only vectors in $\mathcal{F}_k\cap L_k$ are the core vectors in $\mathcal{C}_k$. Therefore, any perfectly balanced sequence in $\mathcal{F}_k$ is formed exclusively by copies of the original dictator blocs.

Conversely, the core sequence
\[
M_{k,1},\dots,M_{k,2k+2}
\]
is contained in $\mathcal{M}_k$ and is perfectly balanced. Since $\mathcal{H}_k=\emptyset$, it witnesses incoherence at length $2k+2$. By \cref{lem:minimality}, no shorter witnessing sequence exists. Hence the incoherence index of $\mathscr{M}_k$ is exactly $2k+2$.
\end{proof}

\begin{cor}[Resolution of Conjecture~5.7 and an Explicit B.25-Type Middle-Layer Construction]\label[cor]{cor:main}
Conjecture~5.7 from \citet{MossPedersen2026} is true: there is no uniform finite bound on the incoherence index of social decision frames. Moreover, for each $k\ge 1$, the construction gives a complement-free family in the middle layer with the balancedness properties required in Conjecture~B.25 for the explicit universe size
\[
2n_k=\binom{2k+2}{k+1}.
\]
\qedthm
\end{cor}

\begin{proof}
We address both conjectures explicitly.

\textbf{Proof of Conjecture 5.7:} We must show that there is no uniform finite bound on the incoherence index of social decision frames. Suppose, for a \emph{reductio ad absurdum}, that there exists a uniform finite upper bound $M \in \mathbb{N}$ on the incoherence index across all frames. By definition, this asserts that every incoherent finite social decision frame must exhibit at least one sequence of subsets of length $m \le M$ that structurally violates the coherence conditions.

By the Archimedean property, we may choose an integer $k \ge 1$ sufficiently large such that $2k+2 > M$. By \cref{thm:main}, there exists a maximal standard frame $\mathscr{M}_k = (W_k, \mathcal{M}_k)$ whose incoherence index is exactly $2k+2$. Because its incoherence index is $2k+2$, the frame $\mathscr{M}_k$ is incoherent.

However, the incoherence index represents the \emph{minimal} length of any sequence in $\mathscr{M}_k$ that violates coherence. Thus, any sequence of subsets in $\mathscr{M}_k$ of length $m < 2k+2$ must perfectly satisfy the coherence conditions. Since $M < 2k+2$, it follows that there is no sequence of length $m \le M$ that violates coherence in $\mathscr{M}_k$. This directly contradicts the assumption that every incoherent frame possesses a coherence violation of length $\le M$. Therefore, no uniform finite bound $M$ can exist.

\textbf{Explicit B.25-type middle-layer construction:} Let
\[
n=n_k=\frac{1}{2}\binom{2k+2}{k+1},
\]
so that $|W_k|=2n_k$. After identifying $W_k$ with $[2n_k]$, the middle layer is precisely the set of all subsets of $W_k$ of size $n_k$. Let
\[
\mathcal{F}^{\mathrm{set}}_k
=
\Bigl\{\,A\subseteq W_k:|A|=n_k \text{ and } \mathbf{x}_A\in \mathcal{F}_k\,\Bigr\}.
\]
We verify the required properties of this family of $n_k$-subsets:
\begin{enumerate}[(1)]
    \item \emph{$\mathcal{F}^{\mathrm{set}}_k$ has no balanced subfamilies of sizes $2,4,\dots,2k$:} Any perfectly balanced sequence of subsets drawn from $\mathcal{F}^{\mathrm{set}}_k$ corresponds algebraically to a zero-sum sequence of vectors in $\mathcal{F}_k$. The proof of \cref{thm:main}, together with \cref{lem:minimality}, shows that any such non-empty sequence has length at least $2k+2$. Therefore, no balanced sequence of length at most $2k$ exists in $\mathcal{F}^{\mathrm{set}}_k$.

    \item \emph{$\mathcal{F}^{\mathrm{set}}_k$ contains a balanced subfamily of size $2k+2$:} The core vectors $\mathcal{C}_k\subseteq\mathcal{F}_k$ correspond precisely to the $2k+2$ dictator blocs $M_{k,j}$. As established in Section~4, this sequence of $2k+2$ subsets covers every voter exactly $k+1$ times, and hence is perfectly balanced.

   \item \emph{$\mathcal{F}^{\mathrm{set}}_k$ is as large as possible among complement-free subfamilies of the middle layer, hence of size $\frac{1}{2}\binom{2n_k}{n_k}$:} The set $X_k$ consists of all $\binom{2n_k}{n_k}$ bipolar indicator vectors corresponding to subsets of $W_k$ of size $n_k$. The vector family $\mathcal{F}_k$ contains exactly one vector from every antipodal pair $\{\mathbf{x},-\mathbf{x}\}$ in $X_k$. Therefore $\mathcal{F}^{\mathrm{set}}_k$ contains no complementary pair of subsets, and
    \[
    |\mathcal{F}^{\mathrm{set}}_k|
    =
    \frac{1}{2}\binom{2n_k}{n_k}.
    \]
\end{enumerate}
This establishes the B.25-type middle-layer properties for the explicit infinite sequence of universe sizes
\[
2n_k=\binom{2k+2}{k+1}.
\]
The form of Conjecture~B.25 that applies to  ``all sufficiently large \(n\)''  requires an additional padding argument.
\end{proof}

The preceding corollary rules out any finite truncation of the coherence
scheme. The next corollary strengthens this to finite axiomatizability in the
full Moss-Pedersen term language: no finite set of sentences, whether or not
drawn from the coherence scheme, defines exactly the measurable frames.

\begin{cor}[Measurable majorities are not finitely axiomatizable]\label[cor]{cor:no_finite_axiomatization}
The class of measurable social decision frames is not finitely axiomatizable in
the term language of \citet{MossPedersen2026}.\qedthm
\end{cor}

\begin{proof}
Suppose, for a contradiction, that there is a finite set $\Gamma$ of sentences
in the Moss-Pedersen language whose finite-frame models are exactly the
measurable social decision frames. Since every sentence in $\Gamma$ is valid on
measurable frames, the completeness theorem of \citet{MossPedersen2026} gives,
for each $\gamma\in\Gamma$, a proof of $\gamma$ from the proof system with its
infinite coherence scheme. Each proof is finite, and $\Gamma$ itself is finite.
Hence only finitely many instances of the coherence scheme occur across all of
these proofs. Let $M$ be the largest sequence length appearing in any of those
instances.

We use the following bounded-soundness observation. Any sentence derivable
using only coherence instances of length at most $M$ is valid in every finite
frame satisfying all coherence instances of length at most $M$, since the
remaining axioms and inference rules of the Moss-Pedersen proof system are
sound on arbitrary finite frames.

Choose $k\ge 1$ with $2k+2>M$. By \cref{thm:main}, there is a maximal standard
frame $\mathscr{M}_k$ whose incoherence index is exactly $2k+2$. Thus
$\mathscr{M}_k$ satisfies every coherence instance of length at most $M$, but
fails coherence at length $2k+2$. In particular, $\mathscr{M}_k$ is not
measurable. By bounded soundness, however, $\mathscr{M}_k$ validates every
sentence in $\Gamma$. This contradicts the assumption that $\Gamma$
axiomatizes exactly the measurable finite frames.
\end{proof}

\section{Conclusion and Future Directions}

This paper has established that there is no uniform finite bound on the incoherence index of social decision frames, thereby resolving Conjecture 5.7 from \citet{MossPedersen2026}. This paper also gives a direct geometric
construction of the middle-layer families predicted by Conjecture~B.25 from \citet{MossPedersen2026}. It does so along
the explicit infinite sequence of universe sizes
\[
2n_k=\binom{2k+2}{k+1}.
\]
Together with the completeness theorem from \citet{MossPedersen2026}, it follows from the
unboundedness of this index that measurable social decision frames admit no finite structural
axiomatization in the Moss-Pedersen language for strict majorities.

The proof-theoretic reason is simple. By completeness, any finite
axiomatization in the Moss-Pedersen language would be derivable using only
finitely many instances of the coherence scheme, and hence only coherence
instances up to some finite sequence length $M$. \cref{thm:main} supplies an
incoherent frame $\mathscr{M}_k$ whose shortest coherence violation has length
$2k+2>M$. That frame satisfies every coherence instance of length at most $M$,
while nevertheless failing representability. Thus, much like the classical
demonstration by \citet{KPS} for comparative probability, measurability for
strict majorities requires an infinite coherence scheme.

The significance is not that strict majority reasoning resists numerical
representation. On the contrary, \citet{MossPedersen2026} show that
representability is exactly characterized by coherence. The present result
shows instead that the boundary between representable and non-representable
majority frames has unbounded finite complexity. Every obstruction is finite,
but there is no finite ceiling on the length of the obstruction required. In
this sense, the measurement-theoretic content of strict majority reasoning is
not exhausted by any finite stock of Moss-Pedersen language conditions.

This result has immediate consequences for the formal logic of majority
reasoning. In \citet{MossPedersen2026}, a natural term logic was introduced to
reason about propositions of the form ``most of everything is an $X$.'' That
logic achieves soundness and completeness by means of an infinite coherence
axiom scheme. The geometric construction presented here shows that this
infinitude is not eliminable: no finite set of sentences in the
Moss-Pedersen language defines exactly the measurable frames.

\begin{enumerate}[\bfseries 1., itemsep=8pt,topsep=10pt]
    \item \textbf{Extremal Bounds on Electorate Size.} There is a substantial
    efficiency gap between the known constructions of highly incoherent frames.
    The geometric proof presented here uses a symmetric dictator core, yielding
    an electorate that scales exponentially with the index
    \[
    |W_k|=\binom{2k+2}{k+1}\sim \mathcal{O}(4^k/\sqrt{k}).
    \]
    In contrast, the Taylor--Zwicker magic-square construction used by
    \citet{Blanco2026}, together with a padding argument, achieves quadratic
    scaling. An open extremal problem is to determine the least electorate size
    $2n$ required to support a maximal standard frame with incoherence index
    $m$.

    \item \textbf{Generalization to Fractional Thresholds.} The geometric
    framework mapping subsets to the rational vector space $\mathbb{Q}^{W_k}$
    naturally generalizes beyond strict majority rule, corresponding to the
    $1/2$ threshold. One might consider super-majoritarian frames, such as
    $2/3$- or $3/4$-threshold frames. The finite-hyperplane avoidance technique
    developed here provides a general algebraic tool for analyzing the linear
    dependencies of these fractional structures, raising the question of
    whether similar unboundedness theorems hold for arbitrary quota rules.
\end{enumerate}

Ultimately, the inherent complexity of strict majority reasoning is not merely
an artifact of measure theory, but a deep combinatorial reality embedded in
the geometry of finite vector spaces.

\newpage
\bibliographystyle{plainnat}
\bibliography{sizelogic}

@article {BB59,
    AUTHOR = {Bia{\l}ynicki-Birula, A. and Browkin, J. and Schinzel, A.},
     TITLE = {On the representation of fields as finite unions of subfields},
   JOURNAL = {Colloq. Math.},
  FJOURNAL = {Colloquium Mathematicum},
    VOLUME = {7},
      YEAR = {1959},
     PAGES = {31--32}
}

@article {TZ95,
    AUTHOR = {Taylor, Alan and Zwicker, William},
     TITLE = {Simple games and magic squares},
   JOURNAL = {J. Combin. Theory Ser. A},
  FJOURNAL = {Journal of Combinatorial Theory. Series A},
    VOLUME = {71},
      YEAR = {1995},
    NUMBER = {1},
     PAGES = {67--88},
}

@article {KPS,
    AUTHOR = {Kraft, Charles H. and Pratt, John W. and Seidenberg, A.},
     TITLE = {Intuitive probability on finite sets},
   JOURNAL = {Ann. Math. Statist.},
  FJOURNAL = {Annals of Mathematical Statistics},
    VOLUME = {30},
      YEAR = {1959},
     PAGES = {408--419}
}

@article{Scott1964,
title = {Measurement structures and linear inequalities},
journal = {Journal of Mathematical Psychology},
volume = {1},
number = {2},
pages = {233-247},
year = {1964},
issn = {0022-2496},
doi = {10.1016/0022-2496(64)90002-1},
url = {https://doi.org/10.1016/0022-2496(64)90002-1},
author = {Dana Scott}
}

@article{Fishburn:1996,
  title={Finite Linear Qualitative Probability},
  author={Fishburn, Peter C},
  journal={Journal of Mathematical Psychology},
  volume={40},
  number={1},
  pages={64--77},
  year={1996}
}

@unpublished{Blanco2026,
  author = {Blanco, Sa{\'u}l A.},
  title  = {On a {Taylor-Zwicker} Construction for Balanced Families and a Conjecture of {Moss} and {Pedersen}},
  note   = {Unpublished manuscript},
  month  = {June},
  year   = {2026}
}

@inproceedings{MossPedersen2026,
  author    = {Moss, Lawrence S. and Pedersen, Arthur Paul},
  title     = {The Measurable Majority},
  booktitle = {Proceedings of the 16th Conference on Logic and the Foundations of Game and Decision Theory ({LOFT} 16)},
  year      = {2026},
  address   = {London, UK},
  note      = {To appear},
  url={
https://doi.org/10.48550/arXiv.2606.23853},
doi={10.48550/arXiv.2606.23853}
}

@article{Narens1980,
  author  = {Narens, Louis},
  title   = {On Qualitative Axiomatizations for Probability Theory},
  journal = {Journal of Philosophical Logic},
  volume  = {9},
  pages   = {143--151},
  year    = {1980}
}

@book{KrantzLuceSuppesTversky1971,
  author    = {Krantz, David H. and Luce, R. Duncan and Suppes, Patrick and Tversky, Amos},
  title     = {Foundations of Measurement, Volume I: Additive and Polynomial Representations},
  publisher = {Academic Press},
  year      = {1971}
}

@book{LuceKrantzSuppesTversky1990,
  author    = {Luce, R. Duncan and Krantz, David H. and Suppes, Patrick and Tversky, Amos},
  title     = {Foundations of Measurement, Volume III: Representation, Axiomatization, and Invariance},
  publisher = {Academic Press},
  year      = {1990}
}

\end{document}